%% file: exponential_service_time_report.tex
\newif\ifreport\reporttrue
\documentclass[conference,10pt,twocolumns]{IEEEtran}

\usepackage{graphicx}
\usepackage[center]{caption}
\usepackage{epsfig,latexsym,amsmath,amsfonts}
\usepackage{amssymb}
\usepackage{placeins}
\usepackage{subfigure}
\usepackage{verbatim}
\usepackage{cite,url}
\usepackage{epsf,bm}
\usepackage{color}
\usepackage{epstopdf}
\usepackage{romannum,hyperref}

\usepackage[utf8]{inputenc}
\usepackage[english]{babel}
\usepackage{amsthm}
\theoremstyle{definition}
\newtheorem{definition}{Definition}

\usepackage[utf8]{inputenc}
\usepackage[english]{babel}

\newtheorem{theorem}{Theorem}
\newtheorem{lemma}[theorem]{Lemma}

\newtheorem{corollary}[theorem]{Corollary}

\usepackage[lined, boxed, linesnumbered, ruled]{algorithm2e}
\begin{document}

\title{Optimizing Data Freshness, Throughput, and Delay in Multi-Server Information-Update Systems}

\title{Optimizing Data Freshness, Throughput, and Delay in Multi-Server Information-Update Systems}
\author{\large Ahmed M. Bedewy$^\dagger$, Yin Sun$^\dagger$, and Ness B. Shroff$^\dagger$$^\ddagger$ \\ [.1in]
\large  \begin{tabular}{c} $^\dagger$Dept. of ECE, $^\ddagger$Dept. of CSE, The Ohio State University, Columbus, OH. \\
\end{tabular} }

\maketitle

\input{sections/abstract}
\input{sections/intro}
\input{sections/sysmodel}
\input{sections/age_optimality}

\input{sections/Numerical_results}
\input{sections/conclusion}
\input{sections/appendices_sec}
\bibliographystyle{IEEEtran}
\bibliography{MyLib}
\end{document}

%% file: sections/abstract.tex
\begin{abstract}
In this work, we investigate the design of information-update systems, where incoming update packets are forwarded to a remote destination through multiple servers (each server can be viewed as a wireless channel). One important performance metric of these systems is the \emph{age-of-information} or simply \emph{age}, which is defined as the time elapsed since the freshest packet at the destination was generated. Recent studies on information-update systems have shown that the age-of-information can be reduced by intelligently dropping stale packets. However, packet dropping may not be appropriate in many applications, such as news and social updates, where users are interested in not just the latest updates, but also past news. Therefore, all packets may need to be successfully delivered. In this paper, we study how to optimize age-of-information without throughput loss. We consider a general scenario where incoming update packets do not necessarily arrive in the order of their generation times. We prove that a preemptive Last Generated First Served (LGFS) policy simultaneous optimizes the age, throughput, and delay performance in infinite buffer queueing systems. We also show age-optimality for the LGFS policy for any finite queue size. These results hold for arbitrary, including non-stationary, arrival processes. To the best of our knowledge, this paper presents the first optimal result on minimizing the age-of-information in communication networks with an external arrival process of information update packets.  
\end{abstract}

%% file: sections/intro.tex
\section{Introduction}\label{Int}
The ubiquity of mobile devices and applications, has increased the demand for real-time information updates, such as news, weather reports, email notifications, stock quotes, social updates, mobile ads, etc. Also, in network-based monitoring and control systems, timely status updates are crucial. These include, but are not limited to, sensor networks used in temperature or other physical phenomenon, and autonomous vehicle systems.

A common objective in these applications is to keep the destination updated with the latest information. To identify the timeliness of the updates, a metric called \emph{age of information}, or simply \emph{age}, was defined in \cite{adelberg1995applying,cho2000synchronizing,golab2009scheduling,KaulYatesGruteser-Infocom2012}. At time $t$, if $U(t)$ is the time when the freshest update at the destination was generated, age $\Delta(t)$  is $\Delta(t)=t-U(t)$. Hence, age is the time elapsed since the freshest packet was generated.

There have been several recent works on characterizing the time-average age of different information-update policies under Poisson arrival process, and finding policies with a small time-average age 
\cite{KaulYatesGruteser-Infocom2012,2012ISIT-YatesKaul,2015ISITHuangModiano,CostaCodreanuEphremides2014ISIT,Icc2015Pappas,
2012CISS-KaulYatesGruteser,KamKompellaEphremides2014ISIT}. In \cite{KaulYatesGruteser-Infocom2012,2012ISIT-YatesKaul,2015ISITHuangModiano}, the update generation rate was optimized to improve data freshness in First-Come First-Served (FCFS) information-update systems. To improve the age, these studies also reduced the update generation rate, which in turn sacrificed the system throughput. In \cite{CostaCodreanuEphremides2014ISIT,Icc2015Pappas}, it was found that the age can be improved by discarding old packets waiting in the queue if a new sample arrives. This can greatly reduce the impact of queueing delay on data freshness. However, many applications may not want to discard packets, e.g., where the users are interested in not just the latest updates, but also past news, in which case all packets must be successfully delivered. In \cite{2012CISS-KaulYatesGruteser,KamKompellaEphremides2014ISIT}, the time-average age was characterized for Last-Come First-Served (LCFS) information-update systems with and without preemption; and FCFS with two servers under a Poisson arrival process. Applications of information updates in  channel information feedback and sensor networks were considered in \cite{ge_info_channel_info_icc2015,BacinogCeranUysal_Biyikoglu2015ITA,2015ISITYates}. 


\begin{figure}
\includegraphics[scale=0.3]{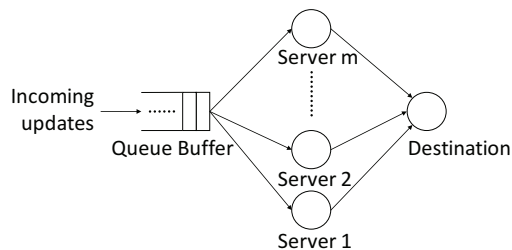}
\centering
\caption{System model.}\label{Fig:sysMod}
\vspace{-0.3cm}
\end{figure}
%
Another important problem is how to maximize data freshness in information-update systems. This involves jointly controlling both the generation and transmission of packet updates \cite{BacinogCeranUysal_Biyikoglu2015ITA,2015ISITYates,generat_at_will}.  An information update policy was developed in \cite{generat_at_will}, which was proven to minimize the time-average age and time-average age penalty among all causally feasible policies. In this setting, a  counter-intuitive phenomenon was revealed:  While a zero-wait  or work-conserving  policy, that generates and submits a fresh update once the server becomes idle, achieves the maximum throughput and the minimum average delay, surprisingly, this zero-wait policy does not always minimize the age. This implies that there is no policy that can simultaneously minimize age and maximize throughput, if the generation and transmission of update packets are jointly controlled.

In this paper, we consider an information-update system which enqueues incoming update packets and forwards them to a remote destination through multiple servers, as shown in Fig. \ref{Fig:sysMod}. In this setting, the updates are generated exogenously to the system, which is different from \cite{generat_at_will}.  We aim to answer the following questions: How to establish age-optimality in a general policy space and under arbitrary arrival process? Is it possible to simultaneously optimize multiple performance metrics, such as age, throughput, and delay? To that end, the following are the key contributions of this paper: 
\begin{itemize}

\item We consider a general scenario where the update packets do not necessarily  arrive in the order of their generation times, which has not been considered before. We prove that, if the packet service times are \emph{i.i.d.} exponentially distributed, then for an arbitrary arrival process and any queue size, a preemptive Last-Generated First-Served (LGFS) policy achieves an age process that is stochastically smaller than any causally feasible policies (Theorem \ref{thm1}). This implies that the preemptive LGFS policy minimizes any non-decreasing functional of the age process. Examples of non-decreasing age penalty functionals include time-average age \cite{KaulYatesGruteser-Infocom2012,2012ISIT-YatesKaul,CostaCodreanuEphremides2014ISIT,Icc2015Pappas,
2012CISS-KaulYatesGruteser,BacinogCeranUysal_Biyikoglu2015ITA,2015ISITYates}, average peak age \cite{2015ISITHuangModiano,CostaCodreanuEphremides2014ISIT,BacinogCeranUysal_Biyikoglu2015ITA}, and time-average age penalty function \cite{generat_at_will}. The intuition is that the freshest update packets are served as early as possible in the preemptive LGFS policy. In particular, the distribution of the age process of the preemptive LGFS policy is invariant over all queue sizes. To the best of our knowledge, this paper presents the first optimal result on minimizing the age-of-information in communication networks with an external arrival process of information update packets.

\item  In addition, we show that if the buffer has an infinite size, then the preemptive LGFS policy is also throughput-optimal and delay-optimal among all causally feasible policies (Theorem \ref{thm2}). \end{itemize}

We note that when the incoming update packets are arriving in the same order of their generation times, the proposed LGFS policy is identical to the LCFS policy studied in \cite{2012CISS-KaulYatesGruteser}. In particular, the time-average age of preemptive and non-preemptive LCFS policies are analyzed in \cite{2012CISS-KaulYatesGruteser} for single-server queueing systems with Poisson arrival process and a queue size of one packet. This paper complements and generalizes the results in \cite{2012CISS-KaulYatesGruteser} by (i) allowing the incoming updates to not arrive in the order of their generation times, (ii) considering more general multi-server queueing systems with arbitrary update arrivals and arbitrary queue size, and (iii) providing an age-optimality proof. 

%% file: sections/sysmodel.tex
\section{Model and Formulation}\label{sysmod}
\subsection{Notations and Definitions}
Throughout this paper, for any random variable $Z$ and an event $A$, let $[Z\vert A]$ denote a random variable with the conditional distribution of $Z$ for given $A$, and $\mathbb{E}[Z\vert A]$ denote the conditional expectation of $Z$ for given $A$.

Let $\mathbf{x}=(x_1,x_2,\ldots,x_n)$ and $\mathbf{y}=(y_1,y_2,\ldots,y_n)$ be two vectors in $\mathbb{R}^n$, then we denote $\mathbf{x}\leq\mathbf{y}$ if $x_i\leq y_i$ for $i=1,2,\ldots,n$. A set $U\subseteq \mathbb{R}^n$ is called upper if $\mathbf{y}\in U$ whenever $\mathbf{y}\geq\mathbf{x}$ and $\mathbf{x}\in U$. We will need the following definitions: 
\begin{definition} \textbf{ Univariate Stochastic Ordering:} \cite{shaked2007stochastic} Let $X$ and $Y$ be two random variables. Then, $X$ is said to be stochastically smaller than $Y$ (denoted as $X\leq_{\text{st}}Y$), if
\begin{equation*}
\begin{split}
\mathbb{P}\{X>x\}\leq \mathbb{P}\{Y>x\}, \quad \forall  x\in \mathbb{R}.
 \end{split}
\end{equation*}
\end{definition}
\begin{definition}\label{def_2} \textbf{Multivariate Stochastic Ordering:} \cite{shaked2007stochastic} 
Let $\mathbf{X}$ and $\mathbf{Y}$ be two random vectors. Then, $\mathbf{X}$ is said to be stochastically smaller than $\mathbf{Y}$ (denoted as $\mathbf{X}\leq_\text{st}\mathbf{Y}$), if
\begin{equation*}
\begin{split}
\mathbb{P}\{\mathbf{X}\in U\}\leq \mathbb{P}\{\mathbf{Y}\in U\}, \quad \text{for all upper sets} \quad U\subseteq \mathbb{R}^n.
 \end{split}
\end{equation*}
\end{definition}
\begin{definition} \textbf{ Stochastic Ordering of Stochastic Processes:} \cite{shaked2007stochastic} Let $\{X(t), t\in [0,\infty)\}$ and $\{Y(t), t\in[0,\infty)\}$ be two stochastic processes. Then, $\{X(t), t\in [0,\infty)\}$ is said to be stochastically smaller than $\{Y(t), t\in [0,\infty)\}$ (denoted by $\{X(t), t\in [0,\infty)\}\leq_\text{st}\{Y(t), t\in [0,\infty)\}$), if, for all choices of an integer $n$ and $t_1<t_2<\ldots<t_n$ in $[0,\infty)$, it holds that
\begin{align}\label{law9'}
\!\!\!(X(t_1),X(t_2),\ldots,X(t_n))\!\leq_\text{st}\!(Y(t_1),Y(t_2),\ldots,Y(t_n)),\!\!
\end{align}
where the multivariate stochastic ordering in \eqref{law9'} was defined in Definition \ref{def_2}.
\end{definition}

\subsection{Queuing System Model}
We consider an information-update system with $m$ identical servers as shown in Fig. \ref{Fig:sysMod}. Each server could be a wireless channel, a TCP connection, etc. The update packets are generated exogenously to the system and then are stored in a queue, waiting to be assigned to one of the servers. Let $B$ denote the buffer size of the queue which can be infinite, finite, or even zero. If $B$ is finite, the queue buffer may overflow and some packets are dropped, which would incur a throughput loss. The packet service times are exponentially distributed with rate $\mu$, which are \emph{i.i.d.} across time and servers.

\subsection{Scheduling Policy}\label{Schpolicy}

The system starts to operate at time $t=0$. A sequence of $n$ update packets are generated at time instants $s_1,\ldots, s_n$, where $n$ can be an arbitrary finite or infinite number, and $0\leq s_1\leq s_2\leq \ldots\leq s_n$. Let $a_i$ be the arrival time of the packet generated at time $s_i$, such that $s_i\leq a_i$. We let $\pi$ denote a scheduling policy that assigns update packets to the servers over time.  The $i$-th generated packet, called packet $i$, completes service at time $c_i$, which depends on the scheduling policy. The packet generation times $(s_1, s_2, \ldots, s_n)$ and packet arrival times $(a_{1}, a_{2}, \ldots, a_{n})$ at the system are arbitrary given, which are independent of the scheduling policy. Note that the update packets may arrive at the system out of the order of their generation times. For example, it may happen that $a_i>a_{i+1}$ but $s_i < s_{i+1}$. 
 
Let $\Pi$  denote the set of all {causal} policies, in which scheduling decisions are made based on the history and current state of the system. We define several types of policies in $\Pi$:

A policy is said to be \textbf{preemptive}, if a server can switch to send any packet at any time; the preempted packets will be stored back into the queue if there is enough buffer space and sent at a later time when the servers are available again. In contrast, in a \emph{non-preemptive} policy, a server must complete delivering the current packet before starting to send another packet.
A policy is said to be \textbf{work-conserving}, if no server is idle when there are packets waiting in the queue.  
 
 
\subsection{Performance Metric}
\begin{figure}
\includegraphics[scale=0.22]{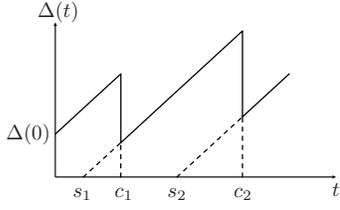}
\centering
\caption{Evolution of the age-of-information $\Delta(t)$.}\label{Fig:Age}
\vspace{-0.3cm}
\end{figure}
Let $U(t)=\max\{s_i : c_{i}\leq t\}$ be the generation time of the freshest packet at the destination at time $t$, where $U(0^-)$ is invariant of the policy $\pi\in\Pi$. The \emph{age-of-information}, or simply the \emph{age}, is defined as
\begin{equation}\label{age}
\begin{split}
\Delta(t)=t-U(t).
\end{split}
\end{equation} 
As shown in Fig. \ref{Fig:Age}, the age increases linearly with $t$ but is reset to a smaller value with the arrival of a fresher packet. The age process is given by 
\begin{align}
\Delta=\{\Delta(t), t\in [0,\infty)\}.
\end{align}
\begin{definition} \textbf{Age Penalty Functional:}
A functional $g(\Delta)$ is said to be an \emph{age penalty functional}, if $g$ is non-decreasing in the following sense: 
\begin{equation}
\begin{split}
&g(\Delta_1) \leq g(\Delta_2),\\
&\text{whenever}\quad \Delta_{1}(t)\leq \Delta_{2}(t),\quad\!\! \forall t\in [0,\infty). 
\end{split}
\end{equation}
\end{definition}
This type of age penalty functional represents the level of ``dissatisfaction'' for data staleness in the network and the ``need'' for fresher information updates. 
Existing examples of  age penalty functionals include:
\begin{itemize}
\item \emph{Time-average age \cite{KaulYatesGruteser-Infocom2012,2012ISIT-YatesKaul,CostaCodreanuEphremides2014ISIT,Icc2015Pappas,
2012CISS-KaulYatesGruteser,BacinogCeranUysal_Biyikoglu2015ITA,2015ISITYates}:} The time-average age is defined as
\begin{equation}\label{functional1}
g_1(\Delta)=\frac{1}{T}\int_{0}^{T} \Delta(t) dt,
\end{equation}
\item \emph{Average peak age \cite{2015ISITHuangModiano,CostaCodreanuEphremides2014ISIT,BacinogCeranUysal_Biyikoglu2015ITA}:} The average peak is defined as 
\begin{equation}\label{functional2}
g_2(\Delta)=\frac{1}{K}\sum_{k=1}^{K} A_{k},
\end{equation}
where $A_{k}$ denotes the $k$-th peak value of $\Delta(t)$ since time $t=0$. 
\item \emph{Time-average age penalty function \cite{generat_at_will}:} The time-average age penalty function is defined as
\begin{equation}\label{functional3}
g_3(\Delta)= \frac{1}{T}\int_{0}^{T} h(\Delta(t)) dt,
\end{equation}
where $h$ : $[0,\infty)\to [0,\infty)$ can be any non-negative and non-decreasing function.
\end{itemize}

%% file: sections/age_optimality.tex
\section{Optimality Analysis}\label{proof1} 
 \begin{algorithm}[h]
\SetKwData{NULL}{NULL}
\SetCommentSty{small} 
$\alpha:=0$\;
\While{the system is ON} {
\If{a new packet with generation time $s$ arrives}{ 
\uIf{all servers are busy}{
\uIf{ $s\leq\alpha$}{
Store the packet in the queue\;}
\Else(~~~~~~\tcp*[h]{The packet carries fresh information.}){
The new packet is assigned to a server by preempting the packet with generation time $\alpha$\; 
The preempted packet with generation time $\alpha$ is stored back to the queue\;
Set $\alpha$ as the smallest generation time of the packets under service\;
}}
\Else(~~~~~~\tcp*[h]{At least one of the servers is idle.})
{
Assign the new packet to one idle server\;
Set $\alpha$ as the smallest generation time of the packets under service\;
} 
}

\If{a packet is delivered}{
 \If{ the queue is not empty}{
 Pick the freshest packet in the queue and assign it to the idle server\;
 Set $\alpha$ as the smallest generation time of the packets under service\;
 }
}
}
\caption{Preemptive Last Generated First Served policy.}\label{alg1}
\end{algorithm}
 In this section, we study a LGFS policy, in which the packets under service are generated the latest (i.e., the freshest) among all packets in the queue; after service, the next freshest packet in the queue is assigned to the idle server. The implementation details of a preemptive LGFS (prmp-LGFS) policy is depicted in Algorithm \ref{alg1}, where $\alpha$ is the smallest generation time of the packets under service. 
 
Define a set of parameters $\mathcal{I}=\{n,(s_i, a_{i})_{i=1}^{n}, B\}$, where $n$ is the total number of packets, $s_i$ and $a_{i}$ are the generation time and the arrival time of packet $i$, respectively, and $B$ is the queue buffer size. Let $\Delta_{\pi}=\{\Delta_{\pi}(t), t\in [0,\infty)\}$ be the age processes under policy $\pi$. The age performance of prmp-LGFS policy is provided in the following theorem.
\begin{theorem}\label{thm1}
Suppose that the packet service times are exponentially distributed and \emph{i.i.d.} across time and servers, then for all $\mathcal{I}$ and $\pi\in\Pi$
\begin{equation}\label{thm1eq1}
[\Delta_{\text{prmp-LGFS}}\vert\mathcal{I}]\leq_{\text{st}}[\Delta_{\pi}\vert\mathcal{I}],
\end{equation}
or equivalently, for all $\mathcal{I}$ and non-decreasing functional $g$
\begin{equation}\label{thm1eq2}
\mathbb{E}[g(\Delta_{\text{prmp-LGFS}})\vert\mathcal{I}]=\min_{\pi\in\Pi}\mathbb{E}[g(\Delta_\pi)\vert\mathcal{I}],
\end{equation}
provided the expectations exist.
\end{theorem}

We need to define the system state of any policy $\pi$:\\
\begin{definition}  At any time $t$, the system state of policy $\pi$ is specified by  $\mathbf{V}_{\pi}(t)=(U_{\pi}(t), \alpha_{1,\pi}(t),\ldots,\alpha_{m,\pi}(t)) $, where $U_{\pi}(t)$ is the generation time of the freshest packet that have already been delivered to the destination. Define $\alpha_{i,\pi}(t)$ as the $i$-th largest generation time of the packets being processed by the servers. Without loss of generality, if $k$ servers are sending stale packets (i.e., $\alpha_{m,\pi}(t) \leq \alpha_{(m-1),\pi}(t) \ldots\leq \alpha_{(m-k+1),\pi}(t) \leq U_\pi(t)$) or $k$ servers are idle, then we set $\alpha_{m,\pi}(t) =\ldots =\alpha_{(m-k+1),\pi}(t) = U_\pi(t)$. Hence, 
\begin{align}\label{sys_state_def.}
U_\pi (t)\leq \alpha_{m,\pi}(t)\leq\ldots\leq\alpha_{1,\pi}(t). 
\end{align}
Let $\{\mathbf{V}_{\pi}(t), t\in[0,\infty)\}$ be the state process of policy $\pi$, which is assumed to be right-continuous. For notational simplicity, let policy $P$ represent the prmp-LGFS policy. By the construction of policy $P$, ${\alpha}_{1,P}(t),{\alpha}_{2,P}(t),\ldots,$ ${\alpha}_{m,P}(t)$ are the generation times of $m$ freshest packets among all packets arrived during $[0,t]$.
\end{definition}

The key step in the proof of Theorem \ref{thm1} is the following lemma, where we compare policy $P$ with any work-conserving policy $\pi$.

 \begin{lemma}\label{lem2}
Suppose that $\mathbf{V}_P(0^-)=\mathbf{V}_{\pi}(0^-)$ for all work conserving policies $\pi$, then for all $\mathcal{I}$
\begin{equation}\label{law9}
\begin{split}
 [\{V_P(t),  t\in[0,\infty)\}\vert\mathcal{I}]\!\geq_{\text{st}}\! [\{V_{\pi}(t), t\in[0,\infty)\}\vert\mathcal{I}].
 \end{split}
\end{equation}
\end{lemma}

 We use coupling and forward induction to prove Lemma \ref{lem2}.
For any work-conserving policy $\pi$, suppose that stochastic processes $\widetilde{\mathbf{V}}_{P}(t)$ and $\widetilde{\mathbf{V}}_{\pi}(t)$ have the same stochastic laws as $\mathbf{V}_{P}(t)$  and $\mathbf{V}_{\pi}(t)$. 
The state processes $\widetilde{\mathbf{V}}_{P}(t)$ and $\widetilde{\mathbf{V}}_{\pi}(t)$
are coupled in the following manner: If the packet with generation time $\widetilde{\alpha}_{i,P}(t)$ is delivered at time $t$ as $\widetilde{\mathbf{V}}_{P}(t)$ evolves, then the packet with generation time $\widetilde{\alpha}_{i,\pi}(t)$ is delivered at time $t$  as $\widetilde{\mathbf{V}}_{\pi}(t)$ evolves. 
Such a coupling is valid since the service time is exponentially distributed and thus memoryless. Moreover, policy $P$ and policy $\pi$ have identical packet generation times $(s_1, s_2, \ldots, s_n)$ and packet arrival times $(a_1, a_2, \ldots, a_n)$. According to Theorem 6.B.30 in \cite{shaked2007stochastic}, if we can show 
\begin{equation}\label{main_eq}
\begin{split}
\mathbb{P}[\widetilde{\mathbf{V}}_{P}(t)\geq\widetilde{\mathbf{V}}_{\pi}(t), t\in[0,\infty)\vert\mathcal{I}]=1,
\end{split}
\end{equation}
then \eqref{law9} is proven.
To ease the notational burden, we will omit the tildes henceforth on the coupled versions and just use $\mathbf{V}_P(t)$ and $\mathbf{V}_{\pi}(t)$. Next, we use the following lemmas to prove \eqref{main_eq}:
\newtheorem{lem}{Lemma}[section] 
\newtheorem{lemx}{Lemma}
\renewcommand{\thelemx}{2*} 
\begin{lemx}\label{lem3'}
Suppose that the system state of policy $P$ is $\{U_P,  \alpha_{1,P},\ldots, \alpha_{m,P}\}$, and meanwhile the system state of policy $\pi$ is $\{U_{\pi}, \alpha_{1,\pi},\ldots,\alpha_{m,\pi}\}$. If
\begin{equation}\label{hyp1'}
U_P \geq U_{\pi},
\end{equation}
then,
\begin{equation}\label{law6'}
 \alpha_{i,P} \geq \alpha_{i,\pi}, \quad \forall i=1,\ldots,m.
\end{equation}
\end{lemx}

\begin{proof}
See Appendix~\ref{Appendix_A'}.
\end{proof}

\begin{lemma}\label{lem3}
Suppose that under policy $P$, $\{U_P',  \alpha_{1,P}',\ldots, \alpha_{m,P}'\}$ is obtained by delivering a packet with generation time $\alpha_{l,P}$ to the destination in the system whose state is $\{U_P,  \alpha_{1,P},\ldots, \alpha_{m,P}\}$. Further, suppose that under policy $\pi$, $\{U_{\pi}', \alpha_{1,\pi}',\ldots,\alpha_{m,\pi}'\}$ is obtained  by delivering a packet with generation time $\alpha_{l,\pi}$ to the destination in the system whose state is $\{U_{\pi}, \alpha_{1,\pi},\ldots,\alpha_{m,\pi}\}$. If
\begin{equation}\label{hyp1}
\alpha_{i,P} \geq \alpha_{i,\pi}, \quad \forall  i=1,\ldots,m,
\end{equation}
then,
\begin{equation}\label{law6}
U_P' \geq U_{\pi}',  \alpha_{i,P}' \geq \alpha_{i,\pi}', \quad \forall i=1,\ldots,m.
\end{equation}
\end{lemma}

\begin{proof}
See Appendix~\ref{Appendix_A}.
\end{proof}

\begin{lemma}\label{lem4}
Suppose that under policy $P$, $\{U_P', \alpha_{1,P}',\ldots,\alpha_{m,P}'\}$ is obtained by adding a packet with generation time $s$ to the system whose state is $\{U_P, \alpha_{1,P},\ldots,\alpha_{m,P}\}$. Further, suppose that under policy $\pi$, $\{U_{\pi}', \alpha_{1,\pi}',\ldots,\alpha_{m,\pi}'\}$ is obtained by adding a packet with generation time $s$ to the system whose state is $\{U_{\pi}, \alpha_{1,\pi},\ldots,\alpha_{m,\pi}\}$. If
\begin{equation}\label{hyp2}
U_P\geq U_{\pi},
\end{equation}
then
\begin{equation}\label{law2}
U_P' \geq U_{\pi}', \alpha_{i,P}' \geq \alpha_{i,\pi}', \quad \forall i=1,\ldots,m.
\end{equation}
\end{lemma}
\ifreport 
\begin{proof}
See Appendix~\ref{Appendix_B}.
\end{proof}
\fi

\begin{proof}[ Proof of Lemma \ref{lem2}]
For any sample path, we have that $U_P (0^-) = U_{\pi}(0^-)$ and $\alpha_{i,P}(0^-) = \alpha_{i,\pi}(0^-)$ for $i=1,\ldots,m$. This, together with Lemma \ref{lem3} and \ref{lem4},  implies that  
\begin{equation}
\begin{split}
[U_P(t)\vert\mathcal{I}] \geq [U_{\pi}(t)\vert\mathcal{I}], [\alpha_{i,P}(t)\vert\mathcal{I}] \geq [\alpha_{i,\pi}(t)\vert\mathcal{I}],\nonumber
\end{split}
\end{equation}
holds for all $t\in[0,\infty)$ and $i=1,\ldots,m$. Hence, \eqref{main_eq} follows which implies \eqref{law9} by Theorem 6.B.30 in \cite{shaked2007stochastic}.
This completes the proof.
\end{proof}

%
%

\begin{proof}[Proof of Theorem \ref{thm1}]
As a result of Lemma \ref{lem2}, we have
\begin{equation*}
\begin{split}
[\{U_{P}(t),  t\in[0,\infty)\}\vert\mathcal{I}]\geq_{\text{st}} [\{U_{\pi}(t), t\in[0,\infty)\}\vert\mathcal{I}],
 \end{split}
\end{equation*}
holds for all work-conserving policies $\pi$, which implies
\begin{equation}
\begin{split}
[\Delta_{P}\vert\mathcal{I}]\leq_{\text{st}}[\Delta_{\pi}\vert\mathcal{I}],
 \end{split}
\end{equation}
holds for all work-conserving policies $\pi$.

For non-work-conserving policies, since the service times are exponentially distributed and \emph{i.i.d.} across time and servers, service idling only increases the waiting time of the packet in the system. Therefore, the age under non-work-conserving policies will be greater. As a result, we have 
\begin{equation}
\begin{split}
[\Delta_{P}\vert\mathcal{I}]\leq_{\text{st}}[\Delta_{\pi}\vert\mathcal{I}], \forall \pi\in\Pi.\nonumber
 \end{split}
\end{equation}

Finally, \eqref{thm1eq2} follows directly from \eqref{thm1eq1} using the properties of stochastic ordering \cite{shaked2007stochastic}. This completes the proof.
\end{proof}

Theorem \ref{thm1} tells us that for arbitrary number $n$, packet generation times $(s_1, s_2, \ldots, s_n)$ and arrival times $(a_{1}, a_{2}, \ldots, a_{n})$, and buffer size $B$, the prmp-LGFS policy can achieve age-optimality within the policy space $\Pi$. In addition, \eqref{thm1eq2} tells us that the prmp-LGFS policy minimizes any non-decreasing age penalty functional $g$, including the time-average age \eqref{functional1}, average peak age \eqref{functional2}, and average age penalty \eqref{functional3}. 

As a result of Theorem \ref{thm1}, we can deduce the following corollary: 

\begin{corollary}\label{cor'}
Suppose that the packet service times are exponentially distributed and \emph{i.i.d.} across time and servers, then for all $\mathcal{I}$, the age performance of the prmp-LGFS policy remains the same for any queue size $B\geq 0$.
\end{corollary}
\ifreport 
\begin{proof}
From the definition of policy prmp-LGFS, its queue is used to store the preempted packets and outdated arrived packets. Since the delivery of these packets doesn't affect the age process of prmp-LGFS policy,  the age performance of the prmp-LGFS policy is invariant for any queue size $B\geq 0$. This completes the proof.
\end{proof}
\fi

Finally, the delay and throughput optimality of the prmp-LGFS policy is stated as follows:
\begin{theorem}\label{thm2}
Suppose that the packet service times are \emph{i.i.d.} exponentially distributed across time and servers, then for all $\mathcal{I}$ such that $B = \infty$, the prmp-LGFS policy is throughput-optimal and mean-delay-optimal  among all policies in $\Pi$.
\end{theorem}
\ifreport
In particular, any work-conserving policy is throughput-optimal and mean-delay-optimal. The proof details are provided in Appendix~\ref{Appendix_C}.
\fi

%% file: sections/Numerical_results.tex
\section{Numerical Results}\label{Simulations}
We present some numerical results to illustrate the age performance of different  policies. The packet service times are exponentially distributed with mean $1/\mu =1$, which is \emph{i.i.d.} across time and servers. The inter-generation times are \emph{i.i.d.} Erlang-2 distribution with mean $1/\lambda$. The number of servers is $m$. Hence, the traffic intensity is $\rho=\lambda/m\mu$. The queue size is $B$, which is a non-negative integer.

 
\begin{figure}[t]
\centering
\includegraphics[scale=0.35]{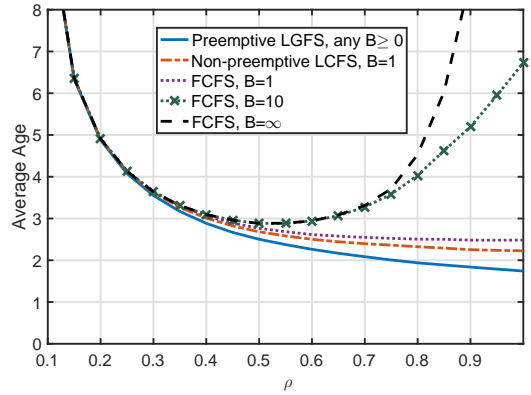}
\caption{Average age versus traffic intensity $\rho$ for an update system with $m=1$ server and queue size $B$.}
\label{avg_age1}
\end{figure}
Figure \ref{avg_age1} illustrates the time-average age versus $\rho$ for an information-update system with $m=1$ server. The time difference between packet generation and arrival ($a_i-s_i$) is zero, i.e., the update packets arrive in the same order of their generation times. One can observe that the preemptive LGFS policy achieves a better (smaller) age than the FCFS policy  analyzed in \cite{KaulYatesGruteser-Infocom2012}, and the non-preemptive LCFS policy with queue size $B=1$ \cite{2012CISS-KaulYatesGruteser} which was also named ``M/M/1/2*'' in \cite{CostaCodreanuEphremides2014ISIT}. 
Note that  in these prior studies, the time-average age was  characterized only for the special case of Poisson arrival process.   Moreover, with ordered arrived packets at the server, the LGFS policy and LCFS policy have the same age performance.


Figure \ref{avg_age2} plots  the time-average age versus $\rho$ for an information-update system with $m=5$ servers. The time difference between packet generation and arrival, i.e., $a_i-s_i$, is modeled to be either  1 or 100, with equal probability.  We found that the age performance of each policy is better than that in Fig. \ref{avg_age1}, because of the diversity provided by five servers. In addition, the preemptive LGFS policy achieves the best age performance among all plotted policies. It is important to emphasize that the age performance of the preemptive LGFS policy remains the same for any queue size $B\geq0$. However, the age performance of the non-preemptive LGFS policy and FCFS policy varies with the queue size $B$ when there are multiple servers. We also observe that the average age in case of FCFS policy with $B=\infty$ blows up when the traffic intensity is high. This is due to the increased congestion in the network which leads to a delivery of stale packets. Moreover, in case of FCFS policy with $B=10$, the average age is high but bounded at high traffic intensity, since the fresh packet has a better opportunity to be delivered in a relatively short period compared with FCFS policy with $B=\infty$. These numerical results validate Theorem \ref{thm1}.


\begin{figure}[t]
\centering
\includegraphics[scale=0.35]{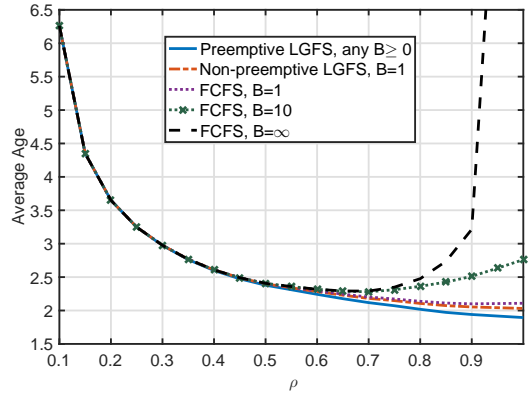}
\caption{Average age versus traffic intensity $\rho$ for an update system with $m=5$ servers and queue size $B$.}
\label{avg_age2}
\end{figure}

%% file: sections/conclusion.tex
\section{Conclusion}\label{Concl}
In this paper, we considered an information-update system, in which update packets are forwarded to a destination through multiple network servers. It was showed that, if the packet service times are \emph{i.i.d.} exponentially distributed, then for any given arrival process and queue size, the preemptive LGFS policy simultaneously optimizes the data freshness, throughput, and delay performance among all causally feasible policies. We will extend these results to more general system settings with general service time distributions.   

%% file: sections/appendices_sec.tex
\appendices

\section{Proof of Lemma \ref{lem3'}}\label{Appendix_A'}

Let $S$ denote the set of packets that have arrived to the system at the considered time epoch. It is important to note that the set $S$ is invariant of the scheduling policy. We use $s_{[i]}$ to denote the $i$-th largest generation time of the packets in $S$. From the definition of the system state and policy $P$, we have
\begin{equation}\label{lnew1}
 \alpha_{i,P}=\max\{s_{[i]},U_P\}, \quad \forall i=1,\ldots,m.
\end{equation}
Since policy $\pi$ is arbitrary policy, the $i$-th freshest packet being processed by the servers under policy $\pi$ is either the $i$-th freshest packet in the set $S$ (the best choice that can be done) or older one. Hence, we have
\begin{equation}\label{lnew2}
 \alpha_{i,\pi}\leq\max\{s_{[i]},U_\pi\}, \quad \forall i=1,\ldots,m,
\end{equation}
where the maximization here follows from the definition of the system state. Since the set $S$ is invariant of the scheduling policy and $U_P \geq U_{\pi}$, this with \eqref{lnew1} and \eqref{lnew2} imply
\begin{equation}
 \alpha_{i,P} \geq \alpha_{i,\pi}, \quad \forall i=1,\ldots,m,
\end{equation}
which completes the proof.

\section{Proof of Lemma \ref{lem3}}\label{Appendix_A}
Since the packet with generation time $\alpha_{l,P}$ is delivered under policy $P$, the packet with generation time $\alpha_{l,\pi}$ is delivered under policy $\pi$, and $\alpha_{l,P}\geq \alpha_{l,\pi}$, we get
\begin{equation}\label{lem3_1}
\begin{split}
 U_P'= \alpha_{l,P}\geq \alpha_{l,\pi}=U_{\pi}'.
 \end{split}
\end{equation}
 From $ U_P'\geq U_{\pi}'$ and using Lemma \ref{lem3'}, we obtain
\begin{equation}\label{lem3_6}
\begin{split}
&\alpha_{i,P}'\geq \alpha_{i,\pi}', \quad  i=1,\ldots,m.
 \end{split}
\end{equation}
Hence, (\ref{law6}) holds for any queue size $B\geq 0$, which completes the proof.

\ifreport 

\section{Proof of Lemma \ref{lem4}}\label{Appendix_B}
Since there is no packet delivery, we have
\begin{equation}\label{law4}
\begin{split}
 U_P'=U_P\geq U_{\pi}=U_{\pi}',
 \end{split}
\end{equation}

From $U_P'\geq U_{\pi}'$ and using Lemma \ref{lem3'}, we obtain
\begin{equation}\label{lem3_6}
\begin{split}
&\alpha_{i,P}'\geq \alpha_{i,\pi}', \quad  i=1,\ldots,m.
 \end{split}
\end{equation} 
Hence, (\ref{law2}) holds for any queue size $B\geq 0$, which completes the proof.

\section{Proof of Theorem \ref{thm2}}\label{Appendix_C}
We follow the same proof technique of Theorem \ref{thm1}. We start by comparing policy $P$ (prmp-LGFS policy) with an arbitrary work-conserving policy $\pi$. For this, we need to define the system state of any policy $\pi$:

\textbf{Definition 8.}  At any time $t$, the system state of policy $\pi$ is specified by  $H_{\pi}(t)=( N_\pi(t), \gamma_\pi(t))$, where $N_\pi(t)$ is the total number of packets in the system at time $t$. Define $\gamma_\pi(t)$ as the total number of packets that are delivered to the destination at time $t$. Let $\{H_{\pi}(t), t\in[0,\infty)\}$ be the state process of policy $\pi$, which is assumed to be right-continuous.

To prove Theorem \ref{thm2}, we will need the following lemma.
\begin{lemma}\label{lem5}
For any work-conserving policy $\pi$, if $H_{P}(0^-)=H_{\pi}(0^-)$ and $B=\infty$, then $[\{{H}_{P}(t),  t\in[0,\infty)\}\vert\mathcal{I}]$ and $[\{H_{\pi}(t), t\in[0,\infty)\}\vert\mathcal{I}]$ are of the same distribution.
\end{lemma}

 Suppose that $\{\widetilde{H}_P(t), t\in[0,\infty)\}$ and $\{\widetilde{H}_{\pi}(t), t\in[0,\infty)\}$ are stochastic processes having the same stochastic laws as $\{H_P(t), t\in[0,\infty)\}$ and $\{H_{\pi}(t), t\in[0,\infty)\}$. Now, we couple  the packet delivery times during the evolution of $\widetilde{H}_P(t)$ to be identical with the packet delivery times during the evolution of $\widetilde{H}_{\pi}(t)$.
 
To ease the notational burden, we will omit the tildes henceforth on the coupled versions and just use $\{H_P(t)\}$ and $\{H_{\pi}(t)\}$. The following two lemmas are needed to prove Lemma \ref{lem5}:

\begin{lemma}\label{lem6}
 Suppose that under policy $P$, $\{N_P', \gamma_P'\}$ is obtained by delivering a packet to the destination in the system whose state is $\{N_P, \gamma_P\}$. Further, suppose that under policy $\pi$, $\{N_\pi', \gamma_\pi'\}$ is obtained  by delivering a packet to the destination in the system whose state is $\{N_\pi, \gamma_\pi\}$. If
\begin{equation*}
N_P = N_{\pi},  \gamma_{P} = \gamma_{\pi},
\end{equation*}
then
\begin{equation}\label{law6_2}
N_P' = N_{\pi}',  \gamma_{P}' = \gamma_{\pi}'.
\end{equation}
\end{lemma}

\begin{proof}
Since there is a packet delivery, we have
\begin{equation*}
\begin{split}
&N_P'=N_P-1=N_\pi-1=N_\pi',\\
&\gamma_P'=\gamma_P+1=\gamma_\pi+1=\gamma_\pi'.
\end{split}
\end{equation*}
Hence, (\ref{law6_2}) holds, which complete the proof.
\end{proof}

\begin{lemma}\label{lem7}
Suppose that under policy $P$, $\{N_P', \gamma_P'\}$ is obtained by adding a new packet to the system whose state is $\{N_P, \gamma_P\}$. Further, suppose that under policy $\pi$, $\{N_\pi', \gamma_\pi'\}$ is obtained by adding a new packet to the system whose state is $\{N_\pi, \gamma_\pi\}$. If
\begin{equation*}
N_P = N_{\pi},  \gamma_{P} = \gamma_{\pi},
\end{equation*}
then
\begin{equation}\label{law7_2}
N_P' = N_{\pi}',  \gamma_{P}' = \gamma_{\pi}'.
\end{equation}
\end{lemma}

\begin{proof}
Because $B=\infty$, no packet is dropped in policy $P$ and policy $\pi$.
Since there is a new added packet to the system, we have
\begin{equation*}
\begin{split}
N_P'=N_P+1=N_\pi+1=N_\pi'.
\end{split}
\end{equation*}
Also, there is no packet delivery, hence 
\begin{equation*}
\begin{split}
\gamma_P'=\gamma_P=\gamma_\pi=\gamma_\pi'.
\end{split}
\end{equation*}
Thus, (\ref{law7_2}) holds, which complete the proof.
\end{proof}

\begin{proof}[ Proof of Lemma \ref{lem5}] 
For any sample path, we have that $N_P (0^-) = N_{\pi}(0^-)$ and $\gamma_{P}(0^-) = \gamma_{\pi}(0^-)$. This, together with Lemma \ref{lem6} and \ref{lem7}, implies that
\begin{equation}
\begin{split}
[N_P(t)\vert\mathcal{I}] = [N_{\pi}(t)\vert\mathcal{I}], [\gamma_{P}(t)\vert\mathcal{I}] = [\gamma_{\pi}(t)\vert\mathcal{I}], \nonumber
\end{split}
\end{equation}
holds for all $t\in[0,\infty)$. This implies that $[\{{H}_P(t),  t\in[0,\infty)\}\vert\mathcal{I}]$ and $[\{H_{\pi}(t), t\in[0,\infty)\}\vert\mathcal{I}]$ are of the same distribution, which completes the proof.
%
%
\end{proof}

\begin{proof}[Proof of Theorem \ref{thm2}]
As a result of Lemma \ref{lem5}, $[\{\gamma_{P}(t),  t\in[0,\infty)\}\vert\mathcal{I}]$ and $[\{\gamma_\pi(t), t\in[0,\infty)\}\vert\mathcal{I}]$ are of the same distribution for any work-conserving policy $\pi$.
This implies that all work conserving policies have the same throughput performance. Also, from Lemma \ref{lem5}, we have that $[\{N_{P}(t),  t\in[0,\infty)\}\vert\mathcal{I}]$ and $[\{N_\pi(t), t\in[0,\infty)\}\vert\mathcal{I}]$ are of the same distribution for any work-conserving policy $\pi$.
Hence, all work-conserving policies have the same mean-delay performance.

Finally, since the service times are \emph{i.i.d.} across time and servers, service idling only increases the waiting time of the packet in the system. Therefore, the throughput and mean-delay performance under non-work-conserving policies will be worse. As a result, the preemptive LGFS policy  is throughput-optimal and mean-delay-optimal  among all policies in $\Pi$ (indeed, any work-conserving policy with infinite queue size $B=\infty$ is throughput-optimal and mean-delay-optimal).
\end{proof}
\fi